\documentclass{elsart}

\usepackage{amssymb}
\usepackage[centertags]{amsmath}
\usepackage[final]{pdfpages}
\usepackage{multicol}
\makeatother

\newtheorem{theorem}{Theorem}[section]
\newtheorem{lemma}[theorem]{Lemma}

\newtheorem{definition}[theorem]{Definition}

\newenvironment{proof}{\removelastskip\par\medskip}
{\penalty-20\null\hfill$\square$\par\medbreak}

\begin{document}

\begin{frontmatter}



\title{ Comparison Between the Two Models:  New Approach Using the $\alpha-$Divergence }
\thanks[talk]{1.}

\author{\small{DHAKER Hamza, Papa Ngom }} 
\author{\small{and Pierre MENDY}}

\address{LMDAN-Laboratoires de Math\'ematiques de la D\'ecision et d'Analyse Num\'erique \\
 Universit\'e Cheikh Anta Diop BP 5005 Dakar-Fann S\'en\'egal\\
  e-mail : papa.ngom@ucad.edu.sn }

\begin{abstract}
\par We propose new nonparametric accordance R\'enyi-$\alpha$ and $\alpha$-Tsallis divergence estimators for continuous distributions. We discuss this approach with a view to the selection model (on alétoire and autoregressive AR (1)). We lestimateur used by kernel density esttimer underlying. Nevertheless, we are able to prove that the estimators are consistent under certain conditions. We also describe how to apply these estimators and demonstrate their effectiveness through numerical experiments.
  
\end{abstract}

\end{frontmatter}
\section{Introduction}

Many statistical, artificial intelligence, and machine learning problems require efficient estimation of the divergence between two distributions. We assume that these distributions are not given explicitly. Only two finite, independent and identically distributed $(i.i.d.)$ samples are given from the two underlying distributions. The R\'enyi$-\alpha$ (R\'enyi, 1961, 1970) and Tsallis$-\alpha$ (Villmann and Haase, 2010) divergences are two widely applied and prominent examples of probability divergences. The popular Kullback–Leibler $(kl)$ divergence is a special case of these families, and they can also be related to the Csisz\'ar's$-f$ divergence (Csisz\'ar, 1967). Under certain conditions, these divergences can estimate
entropy and can also be applied to estimate R\'enyi and Tsallis mutual information. For more examples and other possible applications of these divergences, see the extended technical report (P\'oczos and Schneider, 2011). Despite their wide applicability, there is no known direct, consistent estimator for R\'enyi$-\alpha$ or Tsallis$-\alpha$ divergence.\\
The closest existing work most relevant to the topic of this paper is the work of Marriott and Newbold (1998) address the problem of Bayesian test of the unit root problem as a Bayesian selection between two models: the random walk model and the model stationary. Lynda and Hocine (2010) use the same approach as Marriott and Newblod (1998). This approach has been recently used in 2010 by  Lynda and Hocine. In this paper we propose another method using the $\alpha$-divergence which produces results very close to the work of Lynda and Hocine (2010) and Marriot and Newbold.
\\

The remainder of this paper is organized as follows. Section 2 introduces the
notation and basic definitions. Section 3 study the asymptotic behavior of the
estimator of $\alpha-$divergence. In Section 4, some applications to test hypotheses
are proposed. Section 5 presents some simulation results. Section 6 concludes
the paper.


\section{Definitions and method}
Marriott and Newbold (1998 ) discuss the Bayesian test of the unit root as follows:
$H_{0}: \phi =1  \quad  vs. \quad  H_{1}: \phi < 1 $ \\
in the model $AR(1)$ with intercept $ X_{t} -\mu = \phi (X_{t-1}- \mu ) + \varepsilon_{t}$
where $\varepsilon_{t}$ are $i.i.d$ $N(0,\sigma^{2})$ and $\mu $ is an unknown parameter.\\
Marriott and Newbold (1998) propose to eliminate the parameter $\mu $ considering the sample $(W_{1},..., W_{n})$ zero mean instead of the sample $(X_{1}, ..., X_{n})$ and
\begin{center}
$W_{t}= X_{t}-X_{t-1}, \quad \quad \forall t=1,...,n. $
\end{center} 

The authors then transform the problem of test, a comparison between the two model, following the Bayesian approach:
\begin{center}
$M_{1}: \quad \quad \quad \quad \quad W_{t}=\varepsilon_{t} \quad \quad \quad$
\end{center}

\begin{center}
$ M_{2}: \quad  W_{t}-\phi W_{t-1} =\varepsilon_{t} - \varepsilon_{t-1} $
\end{center}
According to their approach , we propose a new approach using the $\alpha$-divergence \\
\subsection{Distribution functions of the models} 
 \par - $f$ is the unknown true density of the sample $(W_{1},..., W_{n})$\\
 For all $x\in R$, the kernel estimator we denote $f_{n}(x)$ of $f(x)$
(see, e.g., Watson and Leadbetter (1964a), Watson and Leadbetter (1964b), Foldes and Rejto(1981), Tanner and Wong (1983), Winter
(1987), Diehl and Stute (1988) and Deheuvels and Einmahl (1996, 2000)) is
given by

\begin{equation*}
\widehat{f}_{n}(x)= \sum_{i=0}^{n}\dfrac{1}{h} {\cal K}(\dfrac{W_{i}-x}{h})dF(x) 
\end{equation*}
A kernel ${\cal K}$ will be any measurable function fulfilling the following conditions.\\ 
\textbf{K.1} ${\cal K}(.)$ is of bounded variation on $\mathbb{R}^{n} $\\
\textbf{K.2} ${\cal K}(.)$ is right continuous on $\mathbb{R}^{n} $\\
\textbf{K.3} $\parallel {\cal K} \parallel_{\infty}= \sup_{t\in \mathbb{R}^{n} } \mid {\cal K}(t) \mid = k < \infty $\\
\textbf{K.4} $\int_{\mathbb{R}^{n}}{\cal K}(t)dt=1$\\

- In the model $ M_{1} $, the distribution function $W_{t}$: $f_{1}(W)=\dfrac{1}{\sqrt{2\pi \sigma^{2}}}exp\lbrace -\dfrac{W^{2}}{2\sigma^{2}} \rbrace$
 
- In the model $ M_{2} $, the distribution function $W_{t}$: $f_{2}(W)=\dfrac{1}{\sqrt{\frac{4\pi \sigma^{2}}{1-\phi^{2}}}}exp\lbrace -\dfrac{W^{2}}{4\frac{\sigma^{2}}{1-\phi^{2}}} \rbrace$

\subsection{Divergences}
For the remainder of this work we will assume that ${\cal M}_{0} \subset \mathbb{R}^{d} $
is a measurable set with respect to the $d-$dimensional Lebesgue measure and that $p$ and $q$ are
densities on this domain. The set where they are strictly
positive will be denoted by $supp(p)$ and $supp(q)$, respectively.\\
Let $p$ and $q$ be $\mathbb{R}^{d} \supseteq {\cal M}_{0} \rightarrow \mathbb{R}^{d} $ density functions, and let $\alpha \in \mathbb{R}\setminus \lbrace 0,1\rbrace $. The $\alpha -$divergence ${\cal D} _{\alpha}(p,q)$

\begin{equation*}
{\cal D}_{\alpha}(p,q)=\frac{1}{\alpha (1-\alpha )}\lbrace 1- \int_{\chi}(\dfrac{p_{1}(x)}{q(x)})^{\alpha}q(x)\lambda(dx)\rbrace \quad \quad \quad (1)
\end{equation*}
assuming this integral exists. One can see that this is a special case of  Csiszár's $f$-divergence (Csiszár, 1967) and hence it is always nonnegative. Closely related divergences (but not special cases) to (1) are the R\'enyi$-\alpha$ (R\'enyi, 1961) and the Tsallis-$\alpha$ (Villmann and Haase, 2010) divergences.

\section{Distance between the models of density and true density}
The $\alpha$-divergence between $f$ and $f_j$ is:
\begin{equation*}
{\cal D}_{\alpha}(f,f_{j})=\frac{1}{\alpha (1-\alpha )}\lbrace 1- \int(\dfrac{f(x)}{f_{j}(x)})^{\alpha}f_{j}(x)dx\rbrace \quad j=1,2
 \end{equation*}
we estimate ${\cal D}_{\alpha}(f,f_{j})$ using the representation, (1), by setting 
\begin{equation*}
\widehat{{\cal D}}_{\alpha}(f_{n},f_{j})=\frac{1}{\alpha (1-\alpha )}\lbrace 1- \int(\dfrac{\widehat{f}_{n}}{f_{j}(x)})^{\alpha}f_{j}(x)dx\rbrace \quad j=1,2
 \end{equation*}

\subsection{Properties of the divergence estimator $\widehat{D}_{\alpha}(f,f_{j})$}

\begin{theorem}
Assuming $(K.1)-(K.4)$. Then for each pair of sequence $0< h_{n}^{'} < h\leq h_{n}^{''}$ with $h_{n}^{''} \rightarrow \infty $, $\frac{(nh_{n})^{1-\beta}}{log(n)}\rightarrow \infty $, and $ \frac{\vert \log (h_{n}^{''}) \vert}{\log \log n}\rightarrow \infty $ as $n \rightarrow \infty $, we have:
\begin{equation}
 \displaystyle \lim _{n\rightarrow \infty }\sup_{h_{n}^{'}\leq h \leq h_{n}^{''}}\mid \widehat{{\cal D}}_{\alpha}(f_{n},g)-{\cal D}_{\alpha}(f,g) \mid =0 \quad \quad as 
\end{equation}

\end{theorem}

\begin{lemma}:the Alpha-divergence is closely related to the R\'enyi divergence. We define
an Alpha-divergence as
\begin{equation}
{\cal D}_{\alpha}(f,f_{j})=\dfrac{1}{\alpha(\alpha-1)}\lbrace \e^{\alpha(\alpha -1)R_{\alpha}(f,g)}-1 \rbrace
\end{equation}
 
\end{lemma}

\begin{lemma}:
under the same conditions of the theorem 3.1 we have:
\begin{equation}
\displaystyle \lim _{n\rightarrow \infty } \sup_{h_{n}^{'}\leq h \leq h_{n}^{''}}\vert \widehat{R}_{\alpha}(f_{n},g) - R_{\alpha}(f,g) \vert = 0 \quad \quad a.s
\end{equation}
  
 \end{lemma}
 
\begin{proof}
 \textbf{Proof of Lemma 3.3}:
$$ R_{\alpha}(f,g) =\dfrac{1}{\alpha-1} \log \int f^{\alpha} g^{1-\alpha}dx $$

$$ \widehat{R}_{\alpha}(f_{n},g) =\dfrac{1}{\alpha-1} \log \int f_{n}^{\alpha} g^{1-\alpha}dx $$

$$ \mathbb{\widehat{E}} R_{\alpha}(f_{n},g) =\dfrac{1}{\alpha-1} \log \int \mathbb{E} f^{\alpha} g^{1-\alpha}dx $$

\begin{align}
 \widehat{R}_{\alpha}(f_{n},g)-R_{\alpha}(f,g) & = \widehat{R}_{\alpha}(f_{n},g)-\mathbb{\widehat{E}} R_{\alpha}(f_{n},g)+\mathbb{\widehat{E}} R_{\alpha}(f_{n},g)-R_{\alpha}(f,g) \nonumber \\ & = \Delta_{1} + \Delta_{2}
\end{align}

$\Delta_{1}= \widehat{R}_{\alpha}(f_{n},g)-\mathbb{\widehat{E}} R_{\alpha}(f_{n},g) = \dfrac{1}{\alpha-1} \log\lbrace \dfrac{ \int f_{n}^{\alpha}g^{1-\alpha}dx }{\int \mathbb{E} f_{n}^{\alpha}g^{1-\alpha}dx} \rbrace$

for $z\>0$, $\mid \log z\mid \leq \mid 1-\frac{1}{z} \mid +\mid 1-z \mid$

\begin{align*}
\Delta_{1} & \leq \dfrac{1}{\alpha-1}\lbrace \vert 1- \dfrac{ \int \mathbb{E} f_{n}^{\alpha}g^{1-\alpha}dx }{ \int f_{n}^{\alpha}g^{1-\alpha}dx }\vert +  \vert 1-\dfrac{ \int f_{n}^{\alpha}g^{1-\alpha}dx }{\int \mathbb{E} f_{n}^{\alpha}g^{1-\alpha}dx} \vert\rbrace \\ & \leq  
 \dfrac{1}{\alpha-1}\lbrace  \vert \dfrac{\int f_{n}^{\alpha}g^{1-\alpha}dx - \int \mathbb{E} f_{n}^{\alpha}g^{1-\alpha}dx   }{\int f_{n}^{\alpha}g^{1-\alpha}dx} \vert + \vert \dfrac{\int \mathbb{E} f_{n}^{\alpha}g^{1-\alpha}dx -\int f_{n}^{\alpha}g^{1-\alpha}dx }{\int \mathbb{E} f_{n}^{\alpha}g^{1-\alpha}dx} \vert  \rbrace
\end{align*}

$ \textit{A}_{n}=\lbrace x \setminus f_{n} \geq (nh_{n})^{-\frac{\beta}{2}}\rbrace $\\
$ f^{\alpha}_{n} \geq (nh_{n})^{- \alpha\frac{\beta}{2}} \Rightarrow \int f^{\alpha}_{n}g^{1-\alpha}dx \geq (nh_{n})^{-\alpha \frac{\beta}{2}}\int g^{1-\alpha}dx $

$ (\mathbb{E}f_{n})^{\alpha} \geq (nh_{n})^{-\alpha \frac{\beta}{2}} \Rightarrow \int  (\mathbb{E}f_{n})^{\alpha}g^{1-\alpha}dx \geq (nh_{n})^{-\alpha \frac{\beta}{2}}\int g^{1-\alpha}dx $

\begin{align*}
\Delta_{1} & \leq  \dfrac{2(nh_{n})^{-\alpha \frac{\beta}{2}}}{\alpha-1}\vert \dfrac{ \int f_{n}^{\alpha}g^{1-\alpha}dx-\int \mathbb{E} f_{n}^{\alpha}g^{1-\alpha}dx}{\int g^{1-\alpha}(x)dx} \vert \\ & \leq  \dfrac{2(nh_{n})^{-\alpha \frac{\beta}{2}}}{\alpha-1} sup_{x\in \mathbb{R}^{d} }\mid f_{n}^{\alpha} - (\mathbb{E}f_{n})^{\alpha} \mid 
\end{align*} 

since $h(x)=x$ is a 1-Lipschitz function, for $0< \alpha \leq 1 $ then\\
 $ \mid h(x)^{\alpha}-h(y)^{\alpha} \mid \leq \mid h(x)-h(y)\mid ^{\alpha} $\\
 therefore $0< \alpha\leq 1$ we have $$ \mid h(f_n)^{\alpha}- h(\mathbb{E}f_{n})^{\alpha} \mid \leq \mid f_{n}(x)- \mathbb{E}f_{n}\mid^{\alpha} \Rightarrow \mid f_{n}^{\alpha}(x)- (\mathbb{E}f_{n})^{\alpha}\mid  \leq \mid f_{n}(x)- \mathbb{E}f_{n}\mid^{\alpha}  $$

hence $\Delta_{1} \leq   \dfrac{2(nh_{n})^{-\alpha \frac{\beta}{2}}}{\alpha-1} sup_{x\in \mathbb{R}^{d} }\mid f_{n} - (\mathbb{E}f_{n}) \mid ^{\alpha}$\\

We now impose some slightly more general assumptions on the kernel ${\cal K}(.) $ than that of Lemma3.3. Consider the class of functions $$ {\cal K }(.):= \lbrace {\cal K}(\frac{(x-\cdot)}{h^{\frac{1}{n}}}): h>0, x\in \mathbb{R}^{n} \rbrace $$
For $\varepsilon >0$, set $N(\varepsilon , {\cal K })= \sup _{Q}N(k \varepsilon , {\cal K}, d_{Q})$,  where the supremum is taken over all probability measures $Q$ on $(\mathbb{R}^{n}, {\cal B } )$, where ${\cal B } $ represents the $\sigma$-field of Borel sets of $\mathbb{R}^{d}$. Here, $d_{Q}$ denotes the $L_{2}(Q)$-metric and $N({\cal k } \varepsilon , {\cal K}, d_{Q})$ is the minimal number of balls $\lbrace g: d_{Q}(g,g^{'})<\varepsilon \rbrace$ of $d_{Q}$-raduis $\varepsilon$ needed to cover ${\cal K}$. We assume that ${\cal K}$ satisfies the following uniform entropy condition.\\
\textbf{K.5} for some $C>0$ and $\nu >0$, $ N(\varepsilon ,{\cal K})\leq C\varepsilon ^{-\nu}, 0<\varepsilon <1. $\\
\textbf{K.6} ${\cal K } $ is a pointwise measurable class, that is, there exists a countable sub-class ${\cal K }_{0} $ of ${\cal K } $ such that we can find for any function $g\in {\cal K }$ a  sequence of functions $\lbrace g_{m}: m\geq 1\rbrace$ in ${\cal K}_{0}$ for which $ g_{m}(z)\rightarrow g(z), \quad \quad z\in \mathbb{R}^{n} $

By Theorem 1 of Einmahl and Mason (2005), whenever $K(.)$ is measurable and satisfies (K.3-4-5-6), and when $f(.)$ is bounded, we have for each $c>0$, and for a suitable function $\Sigma (c)$, with probability 1,

$$\displaystyle \limsup_{n\rightarrow \infty} \displaystyle \sup_{cn^{-1}\log n \leq h\leq1} \dfrac{ \sqrt{nh_{n}}\| f_{n}- \mathbb{E}f_{n} \|_{ \infty } }{\sqrt{\log(\frac{1}{h})\vee \log \log n}}= \Sigma (c)< \infty $$
 
\begin{equation}
\displaystyle \limsup_{n\rightarrow \infty} \displaystyle \sup_{h_{n} \leq h\leq1} \dfrac{ \sqrt{(nh_{n})^{\alpha}}\vert \Delta_{1} \vert }{\sqrt{(nh_{n})^{\alpha\beta}} \sqrt{(\log(\frac{1}{h})\vee \log \log n)^{\alpha}}}= 0 
\end{equation}

$\Delta_{2} = \widehat{\mathbb{E}}R_{\alpha}(f_{n},g)-R_{\alpha}(f,g)=\frac{1}{\alpha-1}\log \lbrace \dfrac{\int(\mathbb{E}f_{n})^{\alpha}g^{1-\alpha}(x)dx}{\int f^{\alpha}(x)g^{1-\alpha}(x)dx} \rbrace $ 
We repeat the arguments above with the formal change of $ f_{n}$ by $f$
$$ \Delta_{2}\leq \dfrac{2(nh_{n})^{\alpha\frac{\beta}{2}}}{\alpha-1}\vert \mathbb{E}f_{n}-f(x) \vert^{\alpha} $$
 In the other hand, we know (see, e.g, Einmahl and Mason (2005)), that when the density f(.) is
uniformly Lipschitz and continuous, we have for each \\
 $ \parallel \mathbb{E}f_{n}(x)-f(x) \parallel _{\infty}= O(h_{n}^{\frac{1}{d}}) $\\
Thus, we have $ \displaystyle \lim_{n\rightarrow \infty } \displaystyle \sup_{h^{'}_{n}\leq h\leq h^{''}_{n}} \dfrac{2(nh_{n})^{\alpha\frac{\beta}{2}}}{\alpha-1} \parallel  \mathbb{E}f_{n}(x)-f(x) \parallel^{\alpha}_{\infty }=0  $  
 entails that 
\begin{equation}
  \displaystyle \sup_{h_{n}^{'}\leq h \leq h_{n}^{''}}\parallel \Delta_{2}\parallel = \displaystyle \sup_{h_{n}^{'}\leq h \leq h_{n}^{''}} \parallel \widehat{\mathbb{E}}R_{\alpha}(f_{n},g)-R_{\alpha}(f,g)  \parallel \rightarrow 0 
\end{equation} 

Recalling (3.4), the proof of Lemma 3.3  is completed by combining (3.5) with (3.6).

 \end{proof}
 
\begin{proof}
 \textbf{Proof of Theorem 3.1}:
 \begin{align}
  \widehat{{\cal D}}_{\alpha}(f_{n},g)-{\cal D}_{\alpha}(f,g) & =\dfrac{1}{\alpha(\alpha-1)}\lbrace e^{\alpha(\alpha-1)\widehat{R}_{\alpha}(f_{n},g)}-e^{\alpha(\alpha-1)R_{\alpha}(f,g)} \rbrace \nonumber \\ 
  & = \dfrac{1}{\alpha(\alpha-1)}e^{\alpha(\alpha-1)R_{\alpha}(f,g)}\lbrace e^{\alpha(\alpha-1)(\widehat{R}_{\alpha}(f_{n},g)-R_{\alpha}(f,g))}-1 \rbrace
\end{align}
 by (3.3) in connection with (3.7) imply  $$\displaystyle \lim _{n\rightarrow \infty }\sup_{h_{n}^{'}\leq h \leq h_{n}^{''}}\mid \widehat{{\cal D}}_{\alpha}(f_{n},g)-{\cal D}_{\alpha}(f,g) \mid =0 \quad \quad as $$  
 \end{proof}

 \begin{theorem}(asymptotic normality) \label{My-theo}
 Suppose that ${\cal K}: \mathbb{R} \rightarrow \mathbb{R}_{+} $  is a Lipschitz kernel. Then there exists a
sequence $(h_{n})_{n\in \mathbb{N} }$ such that\\ $h_{_{n}}\searrow 0 \quad as \quad n\rightarrow \infty $
then the asymptotic distribution of the estimator of $ \widehat{{\cal D}}_{\alpha}(\widehat{f}_{n},f_{j})$ is gaussian
 $$ \sqrt{nh_{n}} \lbrace \widehat{{\cal D}}_{\alpha}(\widehat{f}_{n},f_{j})-{\cal D}_{\alpha}(f,f_{j}) \rbrace \rightarrow {\cal N}(0 , \sigma_{j}^{2})   $$
 with $\sigma_{j}^{2} = \frac{1}{(1-\alpha)^{2}}(\dfrac{f_{j}(x)}{f(x)})^{4-4\alpha}f^{2}(x)\sigma^{2}(x)   \quad j=1,2 $
  \end{theorem}
 
 \begin{proof}
 \textbf{Proof}:
$ \widehat{{\cal D}}_{\alpha}(\widehat{f}_{n},f_{j})= \displaystyle \frac{1}{\alpha(1-\alpha)}[1-\displaystyle \int (\frac{\widehat{f}_{n}(x)}{f_{j}(x)})^{\alpha}f_{j}(x)dx] \quad j=1,2 $\\
$ \widehat{{\cal D}}_{\alpha}(\widehat{f}_{n},f_{j})-{\cal D}_{\alpha}(f,f_{j}) = \dfrac{1}{\alpha(\alpha-1)}\lbrace \int (f^{\alpha}(x)-\widehat{f}_{n}^{\alpha}(x))f_{j}^{1-\alpha}(x)dx \rbrace $\\
\begin{align*}
\widehat{f}_{n}^{\alpha}(x) & = (\widehat{f}_{n}(x)-f(x)+f(x))^{\alpha}=f^{\alpha}(x)(1+\dfrac{\widehat{f}_{n}(x)-f(x)}{f(x)})^{\alpha} \\ & \approx f^{\alpha}(x)(1+\alpha\dfrac{\widehat{f}_{n}(x)-f(x)}{f(x)})=f^{\alpha}(x)+\alpha f^{\alpha -1}(x)(\widehat{f}_{n}(x)-f(x)) 
\end{align*}
$ \widehat{{\cal D}}_{\alpha}-{\cal D}_{\alpha} =\dfrac{1}{1-\alpha}\lbrace \int (\widehat{f}_{n}(x)(\dfrac{f_{j}(x)}{f(x)})^{1-\alpha}-f(x)(\dfrac{f_{j}(x)}{f(x)})^{1-\alpha})dx \rbrace $\\
$ \dfrac{\sqrt{nh_{n}}}{\sigma (x)}\lbrace\widehat{f}_{n}-f\rbrace \rightarrow {\cal N}(0,1) \Rightarrow (\dfrac{f_{j}(x)}{f(x)})^{1-\alpha}\dfrac{\sqrt{nh_{n}}}{\sigma (x)}\lbrace\widehat{f}_{n}-f\rbrace \rightarrow {\cal N}(0,(\dfrac{f_{j}(x)}{f(x)})^{2-2\alpha})  $
 after the delta method $ \sqrt{nh_{n}} \lbrace \widehat{{\cal D}}_{\alpha}(\widehat{f}_{n},f_{j})-{\cal D}_{\alpha}(f,f_{j}) \rbrace \rightarrow {\cal N }(0, \frac{1}{(1-\alpha)^{2}}(\dfrac{f_{j}(x)}{f(x)})^{4-4\alpha}f^{2}(x)\sigma^{2}(x))$

\end{proof}

\section{Applications for Testing Hypothesis}
The estimate $\widehat{{\cal D}}_{\alpha}(\widehat{f}_{n},f_{j}) \quad j=1,2$ can be used to perform statistical tests.

\subsection{Test of Goodness-Fit}
For completeness, we look at $\widehat{{\cal D}}_{\alpha}(\widehat{f}_{n},f_{j})$ in the usual way, i.e as a goodness-of-fit statistic. Since $\widehat{{\cal D}}_{\alpha}(\widehat{f}_{n},f_{j})$ is a consistent estimator of ${\cal D}_{\alpha}(f,f_{j})$, the null hypothesis when using the statistic $\widehat{{\cal D}}_{\alpha}(\widehat{f}_{n},f_{j})$ is $H_{0}:{\cal D}_{\alpha}(f,f_{j})=0$\\
Under $H_{1}: {\cal D}_{\alpha}(f,f_{j})\neq 0 $, $\widehat{{\cal D}}_{\alpha}(\widehat{f}_{n},f_{j})$ follows a normal distribution decentered, after theorem3.2  
\\

the tests are defined through the critical region, level asymptotic $\alpha$.
$$\lbrace \widehat{{\cal D}}_{\alpha}(\widehat{f}_{n},f_{j}) \geq  \phi^{-1}(1-\alpha )\sigma \rbrace  $$ 

\subsection{Test for Model Selection}
we define an Divergence Indicator  we define an indicator of divergence $${\cal DI}_{\alpha} = {\cal D}_{\alpha} (f,f_{1})-{\cal D}_{\alpha}(f,f_{2})= {\cal D}_{1} -{\cal D}_{2}$$
the estimator of the indicator of the divergence, is given by $$ {\cal \widehat{DI}}_{\alpha} = \sqrt{nh_{n}}(\widehat{{\cal D}}_{\alpha}(\widehat{f}_{n},f_{1})-\widehat{{\cal D}}_{\alpha}(\widehat{f}_{n},f_{2}))=\sqrt{nh_{n}} (\widehat{{\cal D}}_{1} -\widehat{{\cal D}}_{2}) $$

\begin{definition}\

$H_{0}^{eq}: {\cal DI}_{\alpha}=0$  means that the two models are equivalent \\

$ H_{1}^{M_{1}}: {\cal DI}_{\alpha}<0 $ means that model $M_{1}$ is better than model $M_{2}$\\

$ H_{1}^{M_{2}}: {\cal DI}_{\alpha}>0 $ means that model $M_{2}$ is better than model $M_{1}$
\end{definition}

${\cal \widehat{DI}}_{\alpha}$ converges to zero under the null hypothesis $H_{0}^{eq}$, but converges to a strictly negative or positive constant when $H_{1}^{M_{1}}$ and $H_{1}^{M_{2}}$ holds.\\
These properties actually justify the use of ${\cal \widehat{DI}}_{\alpha}$ as a model selction indicator and common procedure of selecting the model with heighest goodness-of-fit.

\begin{theorem}
Under the assumptions of Theorem \ref{My-theo}\\
1) Under the null hypothesis $H_{0}^{eq}$, ${\cal \widehat{DI}}_{\alpha} \longrightarrow {\cal N} (0, \Gamma) $\\
2) Under the $H_{1}^{M_{1}}$ hypothesis $ {\cal \widehat{DI}}_{\alpha} \longrightarrow -\infty $\\
3) Under the $H_{1}^{M_{2}}$ hypothesis $ {\cal \widehat{DI}}_{\alpha} \longrightarrow +\infty $\\
with $\Gamma = \frac{1}{(1-\alpha)^{2}}[ (\frac{f_{1}}{f})^{1-\alpha}-(\frac{f_{2}}{f})^{1-\alpha}]^{4}f^{2}(x)\sigma^{2}(x) $
\end{theorem}

\begin{proof} \textbf{Proof.}
\begin{align*}
 {\cal \widehat{DI}}_{\alpha} & = \sqrt{nh_{n}}(\widehat{{\cal D}}_{1} -\widehat{{\cal D}}_{2}) \\ & =\sqrt{nh_{n}}\lbrace [ \widehat{{\cal D}}_{1}- {\cal D}_{1}]-[\widehat{{\cal D}}_{2}- {\cal D}_{2}]\rbrace +\sqrt{nh_{n}} \lbrace[{\cal D}_{1}-{\cal D}_{2}] \rbrace \\ & =\sqrt{nh_{n}}\lbrace [ \widehat{{\cal D}}_{1}- {\cal D}_{1}]-[\widehat{{\cal D}}_{2}- {\cal D}_{2}]\rbrace +\sqrt{nh_{n}} {\cal DI}_{\alpha} 
\end{align*}
$\circ$  Under the null hypothesis $H_{0}^{eq}$, we have: ${\cal DI}_{\alpha} =0$
\begin{align*}
 {\cal \widehat{DI}}_{\alpha} & = \sqrt{nh_{n}}\lbrace  \widehat{{\cal D}}_{1}- {\cal D}_{1}\rbrace - \sqrt{nh_{n}}\lbrace \widehat{{\cal D}}_{2}- {\cal D}_{2}\rbrace  \\ & = \dfrac{\sqrt{nh_{n}}}{1-\alpha}\lbrace \int (\widehat{f}_{n}(x)(\dfrac{f_{1}(x)}{f(x)})^{1-\alpha}-f(x)(\dfrac{f_{1}(x)}{f(x)})^{1-\alpha})dx 
  \\ & \quad - \int (\widehat{f}_{n}(x)(\dfrac{f_{2}(x)}{f(x)})^{1-\alpha}-f(x)(\dfrac{f_{2}(x)}{f(x)})^{1-\alpha})dx \rbrace \\ & \quad \dfrac{\sqrt{nh_{n}}}{1-\alpha}\lbrace \int f_{n}[ (\frac{f_{1}}{f})^{1-\alpha}-(\frac{f_{2}}{f})^{1-\alpha}]dx- \int f[ (\frac{f_{1}}{f})^{1-\alpha}-(\frac{f_{2}}{f})^{1-\alpha}]dx \rbrace
\end{align*}

$ \sqrt{nh_{n}}\lbrace f_{n}[ (\frac{f_{1}}{f})^{1-\alpha}-(\frac{f_{2}}{f})^{1-\alpha}] - f[ (\frac{f_{1}}{f})^{1-\alpha}-(\frac{f_{2}}{f})^{1-\alpha}] \rbrace \rightarrow {\cal N}(0,[ (\frac{f_{1}}{f})^{1-\alpha}-(\frac{f_{2}}{f})^{1-\alpha}]^{2}\sigma^{2}(x)) $
after the delta method
$ \widehat{{\cal DI}}_{\alpha} \rightarrow {\cal N}(0,\frac{1}{(1-\alpha)^{2}}[ (\frac{f_{1}}{f})^{1-\alpha}-(\frac{f_{2}}{f})^{1-\alpha}]^{4}f^{2}(x)\sigma^{2}(x) ) $
 
 $\circ$  Under the $H_{1}^{M_{1}}$ hypothesis ${\cal ID}_{\alpha}<0 \Rightarrow \sqrt{nh_{n}}{\cal \widehat{DI}}_{\alpha}\rightarrow- \infty $\\
 $\circ$  Under the $H_{2}^{M_{1}}$ hypothesis $ID_{\alpha}>0 \Rightarrow \sqrt{nh_{n}}{\cal \widehat{DI}}_{\alpha}\rightarrow+ \infty $

\end{proof}

\section{Computational Results}
\subsection{Example}
To  illustrate  the  model  procedure  discussed  in  the  preceding section, we consider an example.
\\

We consider various sets of experiments in which data are generated from the mixture of a Normal ${\cal N}(0,1) $ and Normal ${\cal N}(0,2) $ distribution. Hence  the  DGP  (Data  Generating Process) is generated from $m(\pi)$ with the density 
$$m(\pi )= \pi {\cal N}(0,1)+(1-\pi ){\cal N}(0,2) $$

where $\pi (\pi \in [0,1])$  is specific value to each set of experiments.  In  each  set  of  experiment  several  random  sample are drawn from this mixture of distributions. The sample size varies from 100 to 2000,  and for each sample size the number of replication is 1000. we choose two values of the parameter $\alpha =0.5$ , that corresponds to the $\alpha$-divergence. The  aim is to compare the distance beetween true density and the density ${\cal N}(0,1)$, and  the distance beetween the true density and the density ${\cal N}(0,2)$\\
 We choose different values of $\pi$ which are $0.00, 0.25, 0.43, 0.75, 1.00.$\\
Although  our  proposed  model  selection  procedure does not require that the data generating process belong to either of the competing models, we consider the two limiting cases $\pi = 1.00$ and $\pi = 0.00$ for they correspond to  the  correctly  specified  cases. To  investigate  the  case where both competing models are misspecified but not at equal distance from the DGP, we consider the case $\pi = 0.25$, $\pi = 0.75$ and $\pi = 0.43$ second case is interpreted similarly as a ${\cal N}(0,2)$ slightly contaminated by a ${\cal N}(0,1)$  distribution. The  former  case  correspond  to  a  $DGP$  which is ${\cal N}(0,1)$ but slightly contaminated by a ${\cal N}(0,2)$ distribution. In the last case, $\pi = 0.43$ is the value for which  the  $ \widehat{D}_{\alpha}(\widehat{f}_{n},f_{1})$ and  the  $ \widehat{D}_{\alpha}(\widehat{f}_{n},f_{2})$  family  are  approximatively  at  equal  distance to the mixture $m(\pi)$ according to the $\alpha$divergence with the above cells.\\
Thus, this series of experiments approximates the null hypothesis of our proposed model selection test $ \widehat{DI}_{\alpha}$. The  results  of  our  different  sets  of  experiments are presented in \textbf{Tables 1-5}.
\newpage

\begin{center}
\textbf{Table 1. $DGP = {\cal N} (0, 1)$ } 
\end{center}

\begin{center}
\begin{tabular}{ccccccccc}
\hline $n$ &  & {\small 20} & {\small 100} & {\small 300} &{\small 500} & {\small 1000} & {\small 1500} & {\small 2000} \\ 
\hline $ \widehat{{\cal D}}_{1}$ &  & {\small -0.05}& {\small 0.007}  & {\small -0.0020} &{\small 0.016}  & {\small -0.0039}  &{\small 0.012}   & {\small 0.006} \\ 
       $ \widehat{{\cal D}}_{2} $& & {\small 0.16} & {\small 0.12} &{\small 0.14} & {\small 0.16} &{\small 0.14} &{\small 0.14} & {\small 0.14} \\ 
       {\small${\cal \widehat{DI}}_{\alpha}$} &   & {\small -0.21} &{\small -0.11}& {\small -0.15} &{\small -0.14} &{\small  -0.146} & {\small -0.12} &{\small -0.138} \\ 
               &  {\scriptsize Correct} &{\scriptsize 8.4\%} & {\scriptsize 8\%} &  {\scriptsize 26.4\%} & {\scriptsize 57.8\%} &{\scriptsize 95.6\%} & {\scriptsize 100\%} &{\scriptsize 100\%}  \\
               &  {\scriptsize Indecisive} & {\scriptsize 91.6\%} &{\scriptsize 92\%} & {\scriptsize 73.6\%}& {\scriptsize 42.2\%} & {\scriptsize 4.4\%}&{\scriptsize 0\%} &{\scriptsize 0\%}  \\
               &  {\scriptsize Incorrect} &{\scriptsize 0\%} &{\scriptsize 0\%} & {\scriptsize 0\%}&{\scriptsize 0\%} & {\scriptsize 0\%}&{\scriptsize 0\%} &{\scriptsize 0\%} \\
\hline 
\end{tabular} 
\end{center}

\begin{center}
\textbf{Table 2. $DGP = {\cal N} (0, 2)$}  
\end{center}

\begin{center}
\begin{tabular}{ccccccccc}
\hline $n$ & &{\small 20} & {\small 100} & {\small 300} &{\small 500} & {\small 1000} & {\small 1500} & {\small 2000} \\ 
\hline $\widehat{{\cal D}}_{1}$ & & {\small 0.26} & {\small 0.14}   & {\small 0.22}    & {\small 0.28} & {\small 0.23}  &{\small 0.24} &{\small 0.24} \\ 
       $\widehat{{\cal D}}_{2}$& & {\small -0.039} &{\small  -0.016} &{\small -0.008} &{\small -0.006} &{\small  -0.004} &{\small -0.002}&{\small -0.001} \\ 
       ${\cal \widehat{DI}}_{\alpha}$ & & {\small 0.30} & {\small 0.16}	 &{\small 0.23} &{\small 0.29} &{\small 0.23} &{\small 0.24} &{\small 0.24} \\ 
        &  {\scriptsize Correct} &{\scriptsize 30.8\%} & {\scriptsize 68.4\%} &  {\scriptsize 94.2\%} & {\scriptsize 99\%} &{\scriptsize 100\%} & {\scriptsize 100\%} &{\scriptsize 100\%}  \\
               &  {\scriptsize Indecisive} & {\scriptsize 69\%} &{\scriptsize 31.6\%} & {\scriptsize 5.6\%}& {\scriptsize 1\%} & {\scriptsize 0\%}&{\scriptsize 0\%} &{\scriptsize 0\%}  \\
               &  {\scriptsize Incorrect} &{\scriptsize 0.2\%} &{\scriptsize 0\%} & {\scriptsize 0.2\%}&{\scriptsize 0\%} & {\scriptsize 0\%}&{\scriptsize 0\%} &{\scriptsize 0\%} \\
\hline 
\end{tabular} 
\end{center}

\begin{center}
\textbf{Table 3. $DGP =.75* {\cal N} (0, 1) +.25* {\cal N} (0, 2)$}  
\end{center}

\begin{center}
\begin{tabular}{ccccccccc}
\hline $n$ & &{\small 20} & {\small 100} & {\small 300} &{\small 500} & {\small 1000} & {\small 1500} & {\small 2000} \\ 
\hline $\widehat{{\cal D}}_{1}$ & & {\small -0.014} & {\small 0.015}   &  {\small -0.001 }   & {\small 0.01 }   & {\small -0.002} &{\small  0.01} &{\small 0.01} \\ 
       $\widehat{{\cal D}}_{2} $& & {\small 0.19}& {\small 0.19}  & {\small 0.16} &{\small 0.13} &{\small 0.13} &{\small 0.11} &{\small 0.12} \\ 
       ${\cal \widehat{DI}}_{\alpha}$ & & {\small -0.21} &{\small -0.17}&{\small -0.16} &{\small -0.12} &{\small -0.13} &{\small  -0.1} &{\small -0.11} \\ 
        &  {\scriptsize ${\cal N} (0, 1)$ } &{\scriptsize 1.6\%} & {\scriptsize 5.4\%} &  {\scriptsize 34.4\%} & {\scriptsize 67.4\%} &{\scriptsize 99\%} & {\scriptsize 100\%} &{\scriptsize 100\%}  \\
               &  {\scriptsize Indecisive} & {\scriptsize 98.4\%} &{\scriptsize 94.6\%} & {\scriptsize 64.4\%}& {\scriptsize 32.6\%} & {\scriptsize 1\%}&{\scriptsize 0\%} &{\scriptsize 0\%}  \\
               &  {\scriptsize ${\cal N} (0, 2)$} &{\scriptsize 0\%} &{\scriptsize 0\%} & {\scriptsize 0\%}&{\scriptsize 0\%} & {\scriptsize 0\%}&{\scriptsize 0\%} &{\scriptsize 0\%} \\
\hline 
\end{tabular} 
\end{center}
\newpage

\begin{center}
\textbf{Table 4. $DGP =.43* {\cal N} (0, 1) +.57* {\cal N} (0, 2)$}  
\end{center}

\begin{center}
\begin{tabular}{ccccccccc}
\hline $n$ & &{\small 20} &{\small 100} &{\small 300} &{\small 500} &{\small 1000} & {\small 1500} &{\small 2000} \\ 
\hline $\widehat{{\cal D}}_{1}$ & &{\small 0.1} & {\small  0.05} &{\small 0.04} & {\small 0.05} &{\small 0.04} &{\small 0.053}& {\small 0.043}   \\ 
       $\widehat{{\cal D}}_{2} $ & &{\small 0.08}&{\small 0.02}&{\small 0.06} &{\small 0.04}&{\small 0.05} & {\small  0.056} & {\small 0.058}\\ 
       ${\cal \widehat{DI}}_{\alpha}$ & &{\small 0.02} &{\small 0.03} &{\small -0.02}&{\small 0.01} &{\small -0.01} & {\small -0.002} &{\small -0.01} \\ 
       &  {\scriptsize ${\cal N} (0, 1)$ } &{\scriptsize 1.4\%} & {\scriptsize 0.2\%} &  {\scriptsize 0.2\%} & {\scriptsize 0\%} &{\scriptsize 0\%} & {\scriptsize 0\%} &{\scriptsize 0\%}  \\
               &  {\scriptsize Indecisive} & {\scriptsize 98.4\%} &{\scriptsize 99.8\%} & {\scriptsize 99.8\%}& {\scriptsize 100\%} & {\scriptsize 100\%}&{\scriptsize 100\%} &{\scriptsize 100\%}  \\
               &  {\scriptsize ${\cal N} (0, 2)$} &{\scriptsize 0.2\%} &{\scriptsize 0\%} & {\scriptsize 0\%}&{\scriptsize 0\%} & {\scriptsize 0\%}&{\scriptsize 0\%} &{\scriptsize 0\%} \\
\hline 
\end{tabular} 
\end{center}

\begin{center}
\textbf{Table 5. $DGP =.25* {\cal N} (0, 1) +.75* {\cal N} (0, 2)$}  
\end{center}
\begin{center}
\begin{tabular}{ccccccccc}
\hline $n$ & &{\small 20} &{\small 100} &{\small 300} &{\small 500} &{\small 1000} & {\small 1500} &{\small 2000} \\ 
\hline $\widehat{{\cal D}}_{1}$ & & .69& 0.83   &   1.006    &  0.86    &  1.08   & 1.04   & 0.99   \\ 
       $\widehat{{\cal D}}_{2} $& &-0.024 & 0.039    &  0.02     & 0.06    &  0.05   &  0.046   &   0.06 \\ 
       ${\cal \widehat{DI}}_{\alpha}$ & &0.67 & 0.79	 & 1.04	 & 0.8	 & 1.03  &  0.99  &  0.92 \\ 
&  {\scriptsize ${\cal N} (0, 1)$ } &{\scriptsize 0.6\%} & {\scriptsize 0\%} &  {\scriptsize 0\%} & {\scriptsize 0\%} &{\scriptsize 0\%} & {\scriptsize 0\%} &{\scriptsize 0.1\%}  \\
               &  {\scriptsize Indecisive} & {\scriptsize 21\%} &{\scriptsize 17\%} & {\scriptsize 0.4\%}& {\scriptsize 0.2\%} & {\scriptsize 0.2\%}&{\scriptsize 0.2\%} &{\scriptsize 0.1\%}  \\
               &  {\scriptsize ${\cal N} (0, 2)$} &{\scriptsize 78.4\%} &{\scriptsize 83\%} & {\scriptsize 99.6\%}&{\scriptsize 99.8\%} & {\scriptsize 99.8\%}&{\scriptsize 99.8\%} &{\scriptsize 99.9\%} \\
\hline 
\end{tabular} 
\end{center}
\vspace{0.7cm}

Thus  this  set of  experiments  corresponds  approximatively to the null hypothesis of our proposed model selection  test  ${\cal \widehat{DI}}$ .  The  results  of  our  different  sets  of  experiments are presented in Tables 1-5. The first half of each table gives the distance between the true density $f$ and $f_{1}$ sample take density model 1 $ { \cal D}_{1}$, the distance between $f$ and $f_{2}$ Model 2 $ { \cal D}_{2}$ and the differance between the two distance. The second half of each  table  gives  in percentage the number of times our proposed model selection  procedure  based  on ${\cal \widehat{DI}}$   favors  the model 1, the model 2, and indecisive. The  tests  are  conducted  at  $5\%$  nominal  significance level. In the first two sets of experiments ($\pi =0.00$ and $\pi =1.00$) where one model is correctly specified, we use  the  labels "correct, incorrect" and "indecisive" when a choice is made. The first halves of Tables 1-5 confirm our asymptotic results.\\
In Table 5, we observed a high percentage of bad decisions. This is because both models are now specified incorrectly.
 In  contrast, turning to the second halves of the Tables 1 and 2, we first  note  that  the  percentage  of  correct  choices  using DI statistic steadily increases and ultimately conerges to $100\%$\\
 The preceding  comments for the second halves of table 1 and 2  also apply to the second halves of Tables 3 and 4



\begin{figure}
\centering
\begin{minipage}{.5\textwidth}
  \centering
  \includegraphics[height=6cm,width=\linewidth]{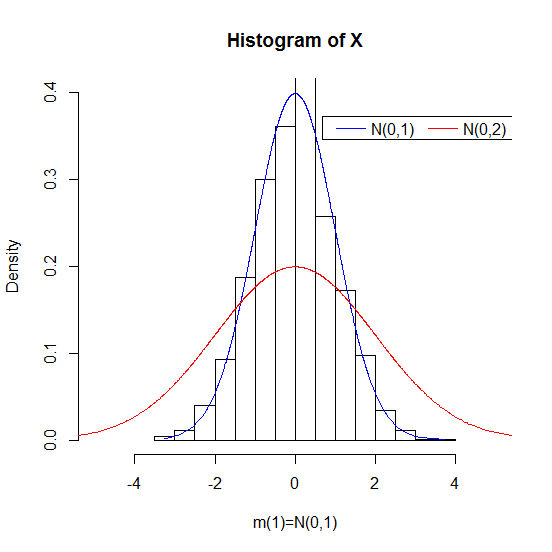}
  \caption{Histogram of ($DGP={\cal N}(0,1)$)}
\end{minipage}%
\begin{minipage}{.5\textwidth}
  \centering
  \includegraphics[height=4cm,width=\linewidth]{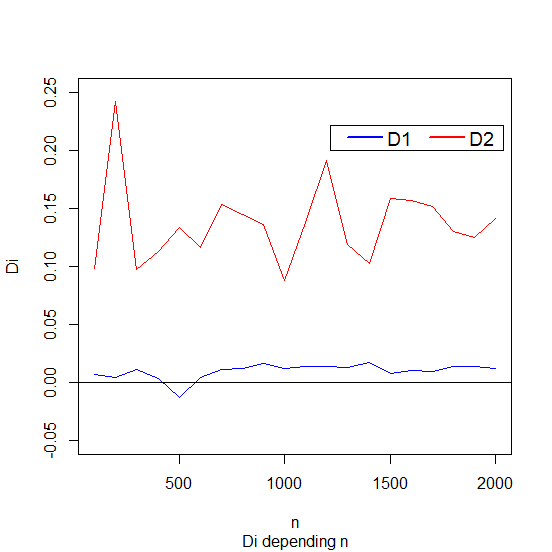}
  \caption{$\widehat{{\cal D}}_{1}$ and $\widehat{{\cal D}}_{2}$ depending $n$}
\end{minipage}
\end{figure}



\begin{figure}
\centering
\begin{minipage}{.5\textwidth}
  \centering
  \includegraphics[height=6cm,width=\linewidth]{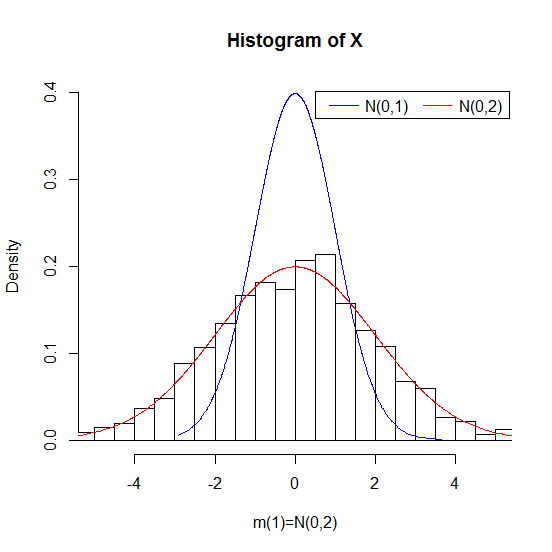}
  \caption{Histogram of ($DGP={\cal N}(0,2)$)}
\end{minipage}%
\begin{minipage}{.5\textwidth}
  \centering
  \includegraphics[height=4cm,width=\linewidth]{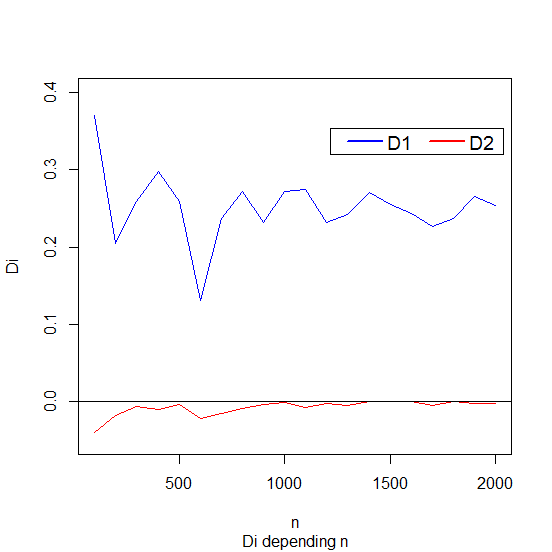}
  \caption{$\widehat{{\cal D}}_{1}$ and $\widehat{{\cal D}}_{2}$ depending $n$}
\end{minipage}
\end{figure}


\begin{figure}
\centering
\begin{minipage}{.5\textwidth}
  \centering
  \includegraphics[height=6cm,width=\linewidth]{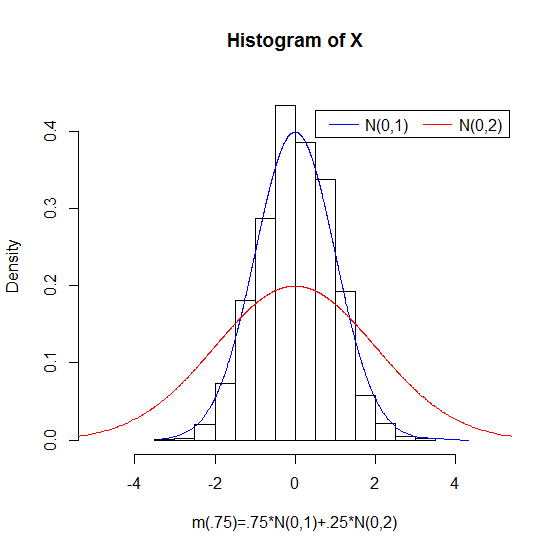}
  \caption{Histogram of ($DGP=.75*{\cal N}(0,1)+.25*{\cal N}(0,2)$)}
\end{minipage}%
\begin{minipage}{.5\textwidth}
  \centering
  \includegraphics[height=4cm,width=\linewidth]{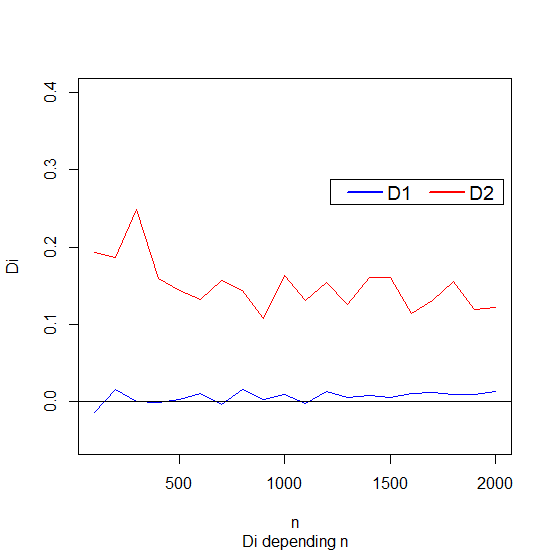}
  \caption{$\widehat{{\cal D}}_{1}$ and $\widehat{{\cal D}}_{2}$ depending $n$}
\end{minipage}
\end{figure}



\begin{figure}[ht]
\centering
\begin{minipage}{.5\textwidth}
  \centering
  \includegraphics[height=6cm,width=\linewidth]{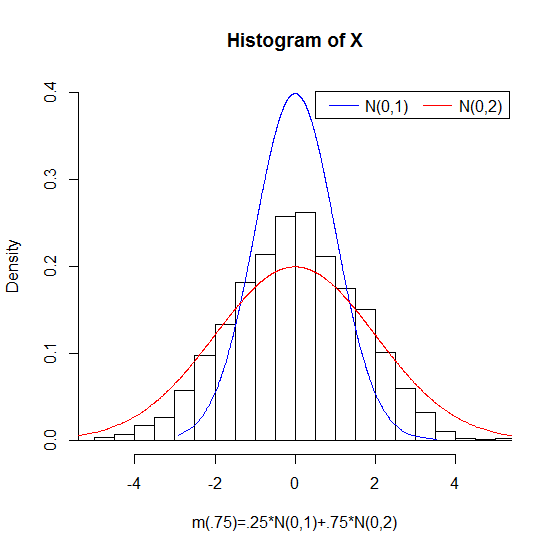}
  \caption{Coparaison barplot of Di depending n ($DGP=.25*{\cal N}(0,1)+.75*{\cal N}(0,2)$)}
\end{minipage}%
\begin{minipage}{.5\textwidth}
  \centering
  \includegraphics[height=4cm,width=\linewidth]{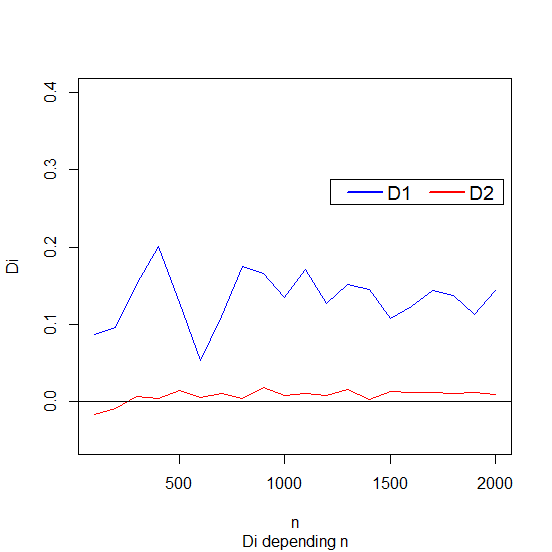}
  \caption{$\widehat{{\cal D}}_{1}$ and $\widehat{{\cal D}}_{2}$ depending $n$}
\end{minipage}
\end{figure}



\begin{figure}
\centering
\begin{minipage}{.5\textwidth}
  \centering
  \includegraphics[height=6cm,width=\linewidth]{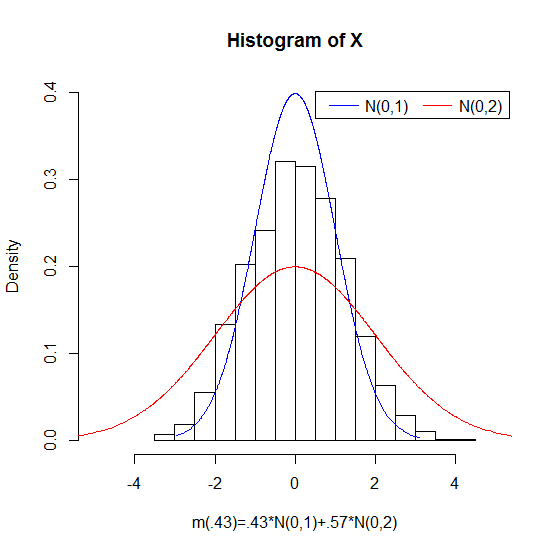}
  \caption{Coparaison barplot of Di depending n ($DGP=.43*{\cal N}(0,1)+.57*{\cal N}(0,2)$)}
\end{minipage}%
\begin{minipage}{.5\textwidth}
  \centering
  \includegraphics[height=4cm,width=\linewidth]{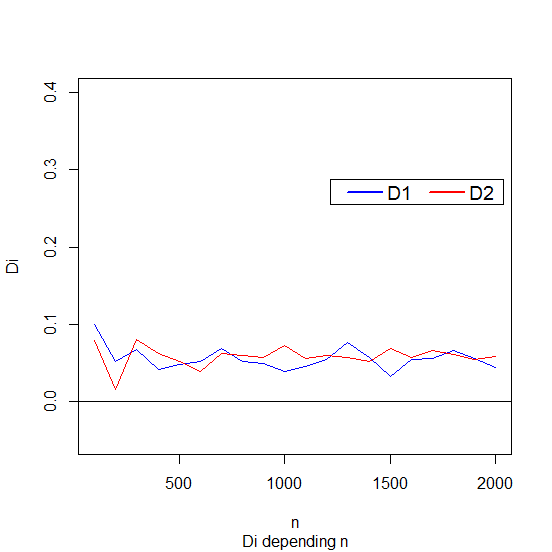}
  \caption{$\widehat{{\cal D}}_{1}$ and $\widehat{{\cal D}}_{2}$ depending $n$}
\end{minipage}
\end{figure}



In \textbf{Figures 1, 3, 5 , 7} and \textbf{9} we plot the histograms of data sets and overlay the curves for $ {\cal N}(0,1) $ and $ {\cal N}(0,2) $ distribution.  When  the  $DGP$  is  correctly  specified \textbf{Figure 1}, the $ {\cal N}(0,1) $ distribution has reasonable chance of being distinguished from $ {\cal N}(0,1) $ distribution.\\
Similarly, in \textbf{Figure 3}, as can be seen, the $ {\cal N}(0,2) $ distribution closely approximates the data sets. 
In \textbf{Figures 5} and \textbf{7} two  distributions  are  close  but  the  ${\cal N}(0,1)$ \textbf{(Figure 5)} and the ${\cal N}(0,2)$ distributions \textbf{(Figure 7)} does appear to be much closer to the data sets.  When $ \pi=0.43$,  the  distribution  for  both  (\textbf{ Figure 9}) $ {\cal N}(0,1) $  distribution  and $ {\cal N}(0,2) $ distribution  are  similar.
\newpage 
 
As expected, our statistic divergence $ \widehat{ {\cal DI } }_{\alpha} $ diverges to $- \infty $ (\textbf{Figures 2} and \textbf{6}) and to $+\infty $ (\textbf{Figures 4} and \textbf{8})  more  rapidly  symmetrical  about  the  axis  that  passes through the mode of data distribution. This follows from the fact that these two distributions are equidistant from the fact that these two distributions are equidistant from the $DGP$ and would be difficult to distinguish from data in practice.\\
 \textbf{Figure 10} allows  a  comparison  with  the  asymptotic  ${\cal N} (0, \Gamma )$  approximation  under  our  null  hypothesis  of equivalence.
\newpage

\section{conclusion}
We learned a new nonparametric estimation for the R\'enyi$-\alpha$ and Tsallis$-\alpha$ divergence, and has been applied to problems of model selection. Under certain conditions, we have shown the consistency of these estimators and how they can be applied to estimate the distance between a known density and an unknown other than estimated by the kernel method. Our tests are based on testing  whether  the  competing  models  are  equally  close to the true distribution against the alternative  hypotheses that one model is closer than the other where closeness of a model is measured according to the discrepancy implicit in the divergence type statistics used. We have also demonstrated their effectiveness by using numerical experiments.

\bibliography{thebibliography}

\end{document}